\def\year{2020}\relax
\documentclass[letterpaper]{article} % DO NOT CHANGE THIS
\usepackage{aaai20}  % DO NOT CHANGE THIS
\usepackage{times}  % DO NOT CHANGE THIS
\usepackage{helvet} % DO NOT CHANGE THIS
\usepackage{courier}  % DO NOT CHANGE THIS
\usepackage[hyphens]{url}  % DO NOT CHANGE THIS
\usepackage{graphicx} % DO NOT CHANGE THIS
\usepackage{nicefrac}
\usepackage{booktabs}
\usepackage{multirow}
\usepackage{algorithm,algorithmicx,algpseudocode}

\usepackage{amsmath}
\usepackage{amsthm}
\usepackage{mathtools}
\usepackage{amsfonts}

\newtheorem{theorem}{Theorem}

\newtheorem{lemma}{Lemma}
\newtheorem{corollary}{Corollary}

\newcommand{\myOmit}[1]{}
\newcommand{\mymax}{\mbox{\rm max}}
\newcommand{\mymin}{\mbox{\rm min}}

\usepackage{color}

\title{Facility Location Problem with Capacity Constraints: Algorithmic and Mechanism Design Perspectives}
\author{\Large \textbf{Haris Aziz,\textsuperscript{\rm 1}  
Hau Chan,\textsuperscript{\rm 2}  
Barton E. Lee,\textsuperscript{\rm 1} 
Bo Li,\textsuperscript{\rm 3}
Toby Walsh\textsuperscript{\rm 4}}\\ 
\textsuperscript{\rm 1} UNSW Sydney and Data61 CSIRO\\ 
\textsuperscript{\rm 2} Department of Computer Science and Engineering, University of Nebraska-Lincoln\\ 
\textsuperscript{\rm 3} Department of Computer Science, University of Oxford\\ 
\textsuperscript{\rm 4} TU Berlin, UNSW Sydney and Data61 CSIRO\\ 
 haris.aziz@unsw.edu.au, hchan3@unl.edu, barton.e.lee@gmail.com, boli@cs.ox.ac.uk, tw@cse.unsw.edu.au
}
 
\begin{document}

\maketitle

\begin{abstract}
We consider the facility location problem in the one-dimensional setting 
where each facility can serve a limited number of agents  from the algorithmic and mechanism design perspectives. 
From the algorithmic perspective, we prove that the corresponding optimization problem, 
where the goal is to locate facilities to minimize either the total cost to all agents or the maximum cost of any agent 
is NP-hard. However, 
we show that the problem is fixed-parameter tractable, and 
the optimal solution can be computed in polynomial time whenever the number of facilities is bounded, or when all facilities have identical capacities. 
We then consider the problem from a mechanism design perspective where the agents are strategic and need not reveal their true locations. 
We show that several natural mechanisms studied in the uncapacitated setting
either lose strategyproofness or a bound on the solution quality %on the returned solution 
for the total or maximum cost objective. 
We then propose new mechanisms that are strategyproof and achieve approximation guarantees that almost match the lower bounds. 
\end{abstract}

\section{Introduction}
%\hau{editing in progress}
In this paper, we study the \emph{facility 
location problem with capacity constraints (FLP-CC)}
from the \emph{algorithmic} and \emph{mechanism design} perspectives. 
In this version of the facility location problem, 
we have a set of agents and a set of facilities,
where each agent is located somewhere on a line, and 
each facility has a capacity limiting the 
number of agents it can serve.
From the \emph{algorithmic perspective}, 
the locations of the agents are \emph{publicly} known, and 
we are interested in the question of 
how to best locate the facilities to minimize the total  
travel distance/cost or the maximum cost of the agents to the located facilities. 
On the other hand, in our \emph{mechanism design} setting, 
the locations of the agents are \emph{privately} known to the agents themselves, and 
our goal is to design mechanisms that elicit the true locations of the agents  
and locate the facilities to minimize the total or maximum cost of the agents subject 
to the agents' reported locations. 
%\emph{Despite the ubiquity of capacity constraints, to the best of our knowledge, the specific 
%FLP-CC considered in the present paper has not previously been studied in the literature.}
%\footnote{Brimberg et al.~\shortcite{BKEM01} studies a similar FLP but assumes that all facilities have identical capacity constraints. Furthermore, the authors focus exclusively on the algorithmic aspects of the problem and do not consider the problem from a mechanism design perspective.}
%Therefore, there is no known or existing algorithmic or mechanism design result for FLP-CC. 
%We consider here the 1-dimensional problem where facilities and
%agents are located along a line

Our FLP-CC models many real-world problems.
The problems include locating schools, hospitals, warehouses, and libraries, 
all of which actually face capacity constraints. 
In the one-dimensional setting, 
our models could be used to describe the setting of locating wastewater plants 
along a river or distribution centres along a highway.
There are also various non-geographical settings that
can be viewed as one-dimensional facility location problems
(e.g. choosing the temperature for a classroom, 
or selecting a committee to represent people with
different political views). In addition, there are
settings where we can use the one-dimensional 
problem to solve more complex problems (e.g. 
decomposing the 2-d rectilinear problem into
a pair of 1-d problems). The one-dimensional problem is also
the starting point to consider more complex metrics (e.g., trees and networks). 

%In many settings, we cannot suppose that
%agents can always be served by the nearest facility. Schools
%have enrollment limits. Hospitals have a fixed number of 
%beds. Warehouses can provide stock for a limited number of stores.
%Libraries serve a maximum number of readers. 

%In our setting of facility location problem with capacity constraints (FLP-CC), 
In FLP-CC, 
we have $n$ agents located on the real line, and we need to locate $m$ facilities 
on the line to serve all the agents\footnote{Our problem
naturally generalizes to higher dimensions as well as
to non-Euclidean distance metrics.}. 
The $i$th facility can serve up to $c_i$ agents.  We assume that $n \leq \sum_{i=1}^m
c_i$ so that every agent can be served.
Each agent $i$ is at location $x_i$, and we suppose that the agents 
are ordered so that $x_1 \leq \ldots \leq x_n$. 
Given an agent $j$, let $a_j\in \{1, \ldots, m\}$ denote the facility that agent $j$ is assigned and let $N_i$ denote the set of agents assigned to facility $i$, i.e.,  $N_i:=\{ j \ | \ a_j = i\}$. A solution is a location $y_i$ for each facility $i$,
and an assignment of agents to facilities such that the capacity constraint $c_i$ of each facility is not exceeded, 
i.e., $|N_i| \le c_i$ for all $i\in \{1, \ldots, m\}$. 
Accordingly, a solution is denoted by $\{(y_j, N_j)\}_{j=1}^m$.
We consider a utilitarian measure: the
\emph{total cost}, $\sum_{j=1}^n |x_j  - y_{a_j}|$;
and an egalitarian measure: the \emph{maximum cost},
$\mymax_{j\in \{1,...,n\}} |x_j  - y_{a_j}|$. 
Our goal is to locate facilities on the line and assign agents
to these facilities to minimize the total or maximum cost. 

\smallskip

\noindent\textbf{Contribution.}
Our main contributions are as follows.
%we introduce a natural and general FLP involving capacity constraints (FLP-CC). 
%which has not previously been studied in the literature. 
Firstly, for FLP-CC,
we provide algorithmic results identifying the complexity of the corresponding optimization problem (NP-hard) but also provide tractability results, via dynamic programs (DPs), for various restricted settings.
Secondly, we explore the mechanism design challenges introduced by this new setting. 
We show that many mechanisms which are considered desirable in uncapitated FLPs become undesirable (with respect to strategyproofness and/or approximation ratio bounds) in FLP-CC. 
%
%For the case of two facilities with capacity constraints, 
%{\color{red} we present new results concerning lower and upper bounds for strategyproof mechanisms. }
%Our lower bound results carry over to the case of three or more facilities.  
We introduce the {\em innerpoint} mechanism which performs relatively well in special cases, 
and characterize this mechanism as the only such strategyproof mechanism within a larger class of mechanisms,
which we call {\em rank} mechanisms. 
Finally, we introduce a new strategyproof mechanism, {\em extended endpoint} mechanism (EEM), which  
achieves approximation guarantees that almost match the lower bounds.
We summarize our contributions in Table \ref{table:results}. 
\begin{table}[htp]
\caption{Algorithmic \& Mechanism Approximation Results}
\begin{center}
\begin{tabular}{c|c|c} \hline
\textbf{Algorithmic Results} & \textbf{Total Cost} & \textbf{Max. Cost} \\ \hline
NP-hard & Exact & Exact \\
DP $O(2^mmn^2)$ & Exact & Exact \\ \hline\hline
\textbf{SP Mechanisms}$^*$ & \textbf{Total Cost} & \textbf{Max. Cost} \\ \hline
Median & Unbounded & Unbounded \\
Endpoint & Unbounded & Unbounded \\
Innerpoint$^1$ & $n/2-1$ & 2 \\
EMM$^2$ & $3n/2$ & 4 \\\hline
\end{tabular}
\end{center}
SP = Strategyproof, *two facilities, $^1$when $c_1=c_2=n/2$, \\ $^2$when $c_1+c_2 \ge n$ 
\label{table:results}
\end{table}%
We leave open the question of upper bounds for three facilities which has remained open since the seminal work~\cite{PrTe09a,ptacmtec2013}, even when all facilities have infinite capacities.

\subsection{Background}
%Below, we review the classic FLP and some recent extensions relevant -- but also distinct -- 
%to the setting we introduce where multiple capacity constrained facilities 
%must be located to service a population of agents.

%~\cite{aclp18}
%
% \barton{I've made an attempt to edit this related literature section. But have kept the original below (commented out). Please feel free to edit/remove as you see fit.}\\
%
% \haris{I likes your discussions in the related work section but can you please try to condense the message so it takes less space?}

\textbf{Algorithmic Perspective.} %FLPs have been studied as optimization problems rather than mechanism design problems which assume strategic %behaviour on the agents part. 
The classic FLP with a single uncapacitated facility can be solved optimally in polynomial time when the objective is to minimize the total or maximum cost (e.g., Procaccia and Tennenholtz~\shortcite{ptacmtec2013}). Similarly, when extending to the setting with multiple uncapacitated facilities, the optimization problem remains tractable admitting a polynomial time solution for either objective function~\cite{MTZC81,MZH83}. 
% If there is just a single capacity constrained facility the problem can be solved optimally in polynomial time when the objective is to minimize the total cost~\cite[Remark 3]{aclp18}.
If there are multiple facilities with identical capacity constraints, the problem is tractable if the objective is to minimize the total cost but becomes intractable (NP-hard) for more general objective functions~\cite{BKEM01}. Complementing these results, we show in the present paper that for our setting, i.e., the multiple capacitated FLP with non-identical capacity constraints, that minimizing either objective function is NP-hard. The intractability result in Brimberg et al.~\shortcite{BKEM01} does not imply intractability in our setting for our objective functions. Moreover, we provide an alternate dynamic program to minimize the total cost for the identical capacity constraint setting studied by Brimberg et al.~\shortcite{BKEM01}.  %A different but related problem is that of placing agents on a line where only one agent can be placed at each slot and agents want to be closest to their target location~\cite{HMO14a,AHMO17a}.
Finally, %approximation algorithms for other more general metric spaces are widely studied in the literature.
we note that our setting is different from the ``capacitated $k$-facility location problem" (see e.g., \cite{pal2001facility,levi2012lp,aardal2015approximation}) where the set of potential facility locations is countable and bounded while, in ours, the set of facility locations is infinite (i.e., the real line).  Thus, we cannot use previous results in the literature directly. 

% there are some differences between our model and the cited model. First, the set of potential facility locations is infinite in our model (i.e., the real line) while the set of potential facility locations in the cited model is countable and bounded. There is no opening cost for the facilities (i.e. opening cost is zero for all), and each client has a unit demand. 

%As a result, we need to introduce new (characterization and algorithmic) techniques to tackle the non-strategic algorithmic version of the problem. We will make this comparison more clear in the paper.
\smallskip

\textbf{Mechanism Design Perspective.} 
%Despite the ubiquity of capacity constraints, to the best of our knowledge the specific facility location problem (FLP) considered in the present paper has not previously been studied in the literature. 
The classic FLP focuses on the setting where a single facility with unlimited capacity, i.e., an uncapacitated facility, must be located along a line with the goal of satisfying certain properties (most importantly strategyproofness) and/or optimizing some objective function. This setting escapes the famous impossibility result of Gibbard-Satterhwaite~\shortcite{gs1,gs2} and inherits many well-known characterization results from social choice, such as those from Moulin~\shortcite{moulin1980}.

%The key result of Black~\shortcite{black1} implies that a mechanism locating the facility at the median of agent locations is strategyproof and minimizes the total cost, i.e, the sum of distances between agents and the facility.

%Similarly the characterisation by Border and Jordan~\shortcite{BoJo83} describes the set of all unanimity respecting and strategyproof mechanisms for the classic FLP.\footnote{Note that efficiency as per Moulin's~\shortcite{moulin1980} characterization result implies that the mechanism is unanimity respecting but the converse does not hold.} 

Two extensions of the classic FLP have attracted attention in the literature:
multiple but uncapacitated facilities (see e.g., Miyagawa  \shortcite{miyagawa2001},
Heo \shortcite{heo2013}, Fotakis and Tzamos \shortcite{ft2013}, and Golowich, Narasimhan, and Parkes~\shortcite{gnpijcai18}), 
and, to a lesser extent, a single capacitated facility.  
%The multiple uncapacitated FLP is studied by Miyagawa  \shortcite{miyagawa2001},
%Heo \shortcite{heo2013}, Fotakis and Tzamos \shortcite{ft2013}, and Golowich, Narasimhan, and Parkes~\shortcite{gnpijcai18}. 
%Due to the absence of capacity constraints, the settings studied by 
%the aforementioned papers allow agents to always be serviced by their nearest facility. 
%This avoids the main strategic tension in the setting we study in this paper. 
Considering the mechanism design problem for facilities with capacity
constraints has only recently been considered by Aziz, Chan, Lee, and
Parkes~\shortcite{aclp18} for a single facility. In their setting,
additional strategic complications arise since some agents will not be
served. The set of agents served by the facility is determined via an equilibrium outcome arising from the induced subgame.\footnote{They show an impossibility result that `reasonable' mechanisms which dictate the agents served can never be strategyproof.}
We do not have such an issue since all our agents can be serviced.
%In the present paper, however, we consider a setting with multiple capacitated facilities such that all agents can be serviced and focus on mechanisms which both locate the facilities and dictate which subsets of agents are serviced by each facility. %(while respecting the capacity constraints of each facility).
Another closely related work is by Procaccia and Tennenholtz~\shortcite{ptacmtec2013}
which considers the problem of designing strategyproof 
mechanisms with approximation guarantees on an objective function (such as total cost and egalitarian welfare) 
for one or two facilities without capacity constraints. 
The works of \cite{ft2013,lu2010asymptotically} prove that any deterministic strategyproof mechanism 
has an approximation ratio of $\Omega(n)$ for the total cost.
Our EEM mechanism has approximation guarantees that almost match the lower bound of the uncapacitated setting.

\section{An Algorithmic Perspective}
%\hau{edit in progress} 
%Agents and facilities are located on $[0,1]$. 
% In our setting of facility location Problem with capacity constraints (FLP-CC), we have $n$ agents located on $[0,1]$, and we need to locate $m$ facilities
% on $[0,1]$ to serve the agents\footnote{Our problem
% naturally generalizes to higher dimensions as well as
% to non-Euclidean distance metrics.}. The $i$th facility can serve up to $c_i$ agents. We assume that $n \leq \sum_{i=1}^m
% c_i$ so that every agent can be served.
% Each agent $i$ is at location $x_i$, and we suppose the agents
% are ordered so that $x_1 \leq \ldots \leq x_n$.
% Our goal is to locate facilities on $[0,1]$ and allocate agents
% to these facilities to minimize the cost.
% A solution is a location $y_i$ for each facility $i$,
% and a facility $a_j \in \{1, ..., m\}$ for each agent $j$ where
% $c_i \geq |\{ j \ | \ a_j = i\}|$.
% We consider an utilitarian measure: the
% \emph{total cost}, $\sum_j |x_j  - y_{a_j}|$.
% We also consider  an egalitarian measure: the \emph{maximum cost},
% $\mymax_j |x_j  - y_{a_j}|$.

We first show that it is intractable 
to find a solution %to the capacitated facility location problem 
that minimizes either the total or maximum cost. This result complements the result in Brimberg et al.~\shortcite[Theorem 2]{BKEM01} 
where intractability is proven for facilities with identical capacity constraints 
and a total cost objective function that can be non-monotonic in the distance between agents and facilities.

\begin{theorem}
Computing a solution that minimizes the total or maximum cost
is NP-hard even when there is no spare capacity in the FLP-CC.
\end{theorem}
\begin{proof}
We show that the problem is 
NP-hard by reducing from the 3-partition problem, 
which is known to be strongly NP-hard \cite{garey}. 
In a 3-partition problem, we are given a 
multiset $T = \{t_1, ..., t_{3n} \}$ of positive integers 
of size $3n$, and we 
want to know if $T$ can be partitioned into 
$n$ subsets $T_1, T_2, ..., T_n$, each of size three, such that 
the sum of the numbers in each subset is the same. 
Let $B = \frac{\sum_{t \in T} t}{n}$. 
In particular, we consider the 3-partition instances in which $B$ and integers  
are polynomially bounded in $n$ and $\frac{B}{4} < t_i < \frac{B}{2}$  for all $i$ 
such that $B$ is a positive integer and the integers are at least 1.  
Take such instance of 3-partition, we reduce it to the decision version 
of FLP-CC where we want to place the facilities to achieve 
the total cost and maximum cost of 0. 
%social welfare of at least of some target $V$ in the following way. 
We let $m = |S| = 3n$ to be the number of facilities, 
and, for each facility $s_j$, $j = 1, ..., 3n$, 
set capacity $k_j  = t_j$. 
With a slight abuse of notation, 
we let $nB = |N|$ to be the number of agents 
such that there are $n$ groups of $B$ agents which 
are equally spaced apart and the agents 
in the same group are located at a single location. 
%Finally, we let $V=nB$. 

Suppose that the instance of the 3-partition problems has a solution. 
It follows that there are $T_1, T_2, ..., T_n$ such that 
each sum up to $B$. For each group $i$ of $B$ agents, 
we place three facilities of capacities $t_j \in T_i$ at the location of group $i$. 
Clearly, we achieve the total cost or maximum cost of 0. 

Suppose that we have a solution to the constructed problem. 
We have a location $a_j$ for each facility $s_j$ such
that the total cost or maximum cost is 0.  
In fact, this is the best cost we can obtain 
where the agents do not have to travel at all.
Moreover, each facility must be located on the location 
of one of the groups and 
all of the $B$ agents within each group can use 
one of the facilities at its location. 
Since the capacity of facility $s_j$ is $\frac{B}{4} < k_j=t_j < \frac{B}{2}$, 
each location of the group has at least three facilities. 
Also since the total number of facilities is $3n$ and 
there are $n$ groups, the number of facilities at each group is exactly 3.  
Finally, the sum of the capacities of the three facilities 
is $B$. Thus, for each group $i$, we 
can construct $T_i$ by take elements with 
the same values as capacities of the facilities located at group $i$. 

To conclude, the problem instance above
is of polynomial size bounded by $n$ and we obtain our claimed result. 
\end{proof}

%\begin{proof}[Proof Sketch]
%We reduce from the 3-partition problem \cite{garey}, 
%which is known to be strongly NP-hard. 
%In this problem, we are given 
%a multiset of $3m$ positive integers with values $d_i$'s 
%and a target integer value $B$. 
%The problem is to find a 3-partition of
%the multiset 
%%a multiset of $3m$ positive integers $d_i$
%into $m$ triplets each with an identical sum $B$. 
%We reduce this to a FLP-CC with optimal
%zero total or maximum cost. We have
%$3m$ facilities, with the $i$th facility
%having capacity $d_i$. We also have
%$mB$ agents on the line with 
%$B$ agents at location $\frac{i}{m}$ for
%$i=1$ to $m$.
%We note that our instance 
%is polynomial size as $B$ is bounded 
%by some polynomial in $m$ (as 3-partition is strongly NP-hard).  
%It can be shown that the optimal solution with zero cost exists if and only if we have a yes instance of 3-partition.
%\end{proof}

We observe that the computational
complexity of finding an optimal solution
does not come from having to decide where
to distribute any spare capacity. 
The reduction uses a FLP-CC instance
where the facilities have a total capacity exactly equal to
the number of agents to be served. The
computational complexity comes from 
having to find the best of the $m!$ possible
left to right orderings of these facilities.

%We offer two methods to deal with this. % computational
                                % intractability. 
%Below, we show that the 
%optimal total cost and maximum cost solutions can be computed 
%in polynomial time via dynamic programs  when the number of facilities is 
%bounded (i.e., fixed-parameter tractable), or when the number of facilities is unbounded
%but they all have the same capacity. 
%Second, we identify some polynomial time mechanisms 
%that return approximately
%optimal solutions. 

% with bounds on the quality of the approximation.

%\subsection{Dynamic programming}

As we show next, to minimize the maximum cost (total cost), we locate each facility in the midpoint (median) of the
continuous region of agents it serves. Thus, for a fixed number of facilities or facilities of identical capacity, 
we can use dynamic programming to compute the optimal (total or maximum cost) solution in polynomial time. 
%If the number of agents is exactly equal to the total capacity, we can enumerate all orderings of facilities, and then easily locate facilities and allocate agents optimally. 

%\footnote{This `continuous region' property of an optimal solution is proven in Brimberg et al.~\shortcite[Theorem 1]{BKEM01} for a setting where facilities have identical capacities. The proof however can be extended easily to the non-identical capacity constraint setting studied in the present paper and with objective function being either the total or maximum cost.} 
%When the number of agents is less than the total capacity, 
%we need to work harder as we don't know how to distribute the spare capacity between facilities. 

%\hau{A quick question: We have a characterization for optimal total cost 
%where you locate the facility at the median of the interval of the people it services; 
%does the same characterization hold for optimal max cost as well? 
%It is worth some space add in such characterization so that it 
%doesn't seem our problem can be directly solve vs DP }

%\hau{Barton can you edit the below? I copied and pasted this from the capacity constrainted RA.pdf; I think before we assume they are bounded from 0 and 1}
%In this subsection we consider how to characterise solutions to the multiple facility problem in the exogenous zoning model, and prove a hardness result for solving the optimisation problem. The characterisation results are utilised in further subsection to provide bounds on the welfare ratio between the two models.

Order the set of agents $N$ such that $x_1\le x_2\le \ldots \le x_{n-1}\le x_n$. We say that a set of agents $N_j$ is \emph{continuous} if $N_j=\{x_{j_1}, x_{j_1+1},\ldots, x_{j_2}\}$ for some $j_1\le j_2$.

\begin{lemma}
For either objective function (total cost or maximum cost), there exists an optimal solution $\{(y_j, N_j)\}_{j=1}^m$   such that each $N_j$ is continuous.
\end{lemma}

\begin{proof}
Suppose that the  solution $\{(y_j, N_j)\}_{j=1}^m$ provides minimal total cost but is such that there exists  $j$ for which $N_j$ is not continuous. That is, the exists some pair of agents $p,s\in N_j$ such that (1) $y_j\le x_p<x_s$, $p\notin N_j$ and $s\in N_j$, or (2) $x_s<x_p\le y_j$, $p\notin N_j$ and $s\in N_j$. Both cases are dealt with similarly and so we assume the first case. 

%If $p\notin \cup_{i=1}^m N_i$ then swapping $p$ and $s$ necessarily leads to a solution with strictly lower total cost and at least as low maximum cost since $d(a_j, x_p)<d(a_j, x_s)$. 

Let $\ell \ : \ p\in N_\ell$.  For any location of $x_s$, if $y_j\le x_p <y_\ell$,  then swapping $p$ from $N_\ell$ and $s$ from $N_j$ leads to  strictly lower total cost. On the other hand, if $y_\ell\le y_j$, then swapping $p$ from $N_\ell$ and $s$ from $N_j$ leads to  no change in total cost. Repeating this method until each set $N_j$ is continuous leads to a solution with at least as low total cost. Since the original solution was optimal, we conclude that there exists an optimal solution with each $N_j$ continuous. 

Now consider the maximum cost objective. A similar argument applies. The only difference is that in the case where $y_j\le x_p <y_\ell$ we can only guarantee weakly lower maximum cost. Nonetheless, the same conclusion is reached --- there exists an optimal maximum cost solution such that each $N_j$ is continuous.
\end{proof}

The above lemma shows that partitions, or sets $N_j$, must be continuous but it does not require that facilities be located `inside' such a partition. However, it is well-known that the median of a set of points minimizes the  sum of  absolute deviations and the midpoint minimizes the maximum deviation. This leads  to the following corollary.  

\begin{corollary}\label{corollary: continuous and median}
An optimal solution $\{(y_j, N_j)\}_{j=1}^m$ is given by by continuous partitions of the agents which correspond to facilities such that each facility $s_j$ locates at the median and midpoint of its associated partition $N_j$ for total cost and maximum cost, respectively. 
\end{corollary}

%\begin{proof}
%Suppose for sake of a contradiction that the optimal solution to the problem is $\{(a_j, N_j)\}_{j=1}^m$ such that there exists $a_j$ which is not equal to the median of $N_j$. Recalling that the median minimises the sum of absolute deviations, it follows immediately that locating $s_j$ at the median can be no worse. Thus, an optimal solution is attained by locating all facilities at the median of their continuous set $N_j$.
%\end{proof}

The above results show that to solve FLP-CC for total and maximum costs, 
it suffices to look for continuous sets which match the capacity of each facility and locate the facility at the median or midpoint of each set.

\subsection{Bounded Number of Facilities}
\paragraph{Total Cost.}
We compute the optimal way (for minimizing the total cost) of locating $m$ facilities 
and partition agents into $m$ parts using dynamic program. 
%First, we introduce some notation. 
We will use $[j,j']$ to denote the set of agents $\{j,\dots,j'\}$ for any $j\le j'$. 
We use $OPT(j,j',k,S)$ to represent the optimal total cost when agents $[1,j']$ are partitioned into $k$ parts, and $j'\ge j$ are in the $k$-th partition, and $S$ is the subset of $k$ facilities used to serve agents $[1,j']$. 
We will make use of a $n\times n\times m\times 2^m$ array $M$, whose entries are initially set to {\em empty}. We will use $M[j,j',k,S]$ to store $OPT(j,j',k,S)$. 
We use $v([j,j'],s)$ to denote the optimal total cost when agents $[j,j']$ are served by facility $s$ (with capacity $c$). In order to compute $v([j,j'],s)$ we can simply try to locate $s$ at each of the agent locations and serve $c$ nearest agents.\footnote{The optimal point will be the median of those subset of agents who are actually served.}

We invoke $OPT^{T}(n,n,m,S)$ (Algorithm~\ref{algo:opt}) to compute the optimal total cost and recover the solution from the values stored in $M$. 

\algrenewcommand\algorithmicindent{1.0em}
\begin{algorithm}
	\caption{$OPT^{T}(j,j',k,S)$}
	\footnotesize
	\begin{algorithmic}[1]
		\If{$j=0$ or $k=0$ or $S=\emptyset$}
		\State \Return $0$
		\ElsIf{$M[j,j',k,S]$ not {\em empty}}
		\State \Return $M[j,j',k,S]$
		\Else
		\State \label{line:part}\noindent$s^*\gets \arg\min_{s\in S}\{OPT^T(j-1,j-1,k-1,S\setminus \{s\}) + v([j,j'],s)\}$
		%\parState{\label{line:rec}\noindent
			\begin{align*}
			&M[j,j',k,S]\gets \min\{ OPT^T(j-1,j',k,S),\\
			&OPT^T(j-1,j-1,k-1,S\setminus \{s^*\}) + v([j,j'],s^*)\}\}
			\end{align*}	
		%}
		\State \Return $M[j,j',k,S]$
		\EndIf
	\end{algorithmic}
	\label{algo:opt}
\end{algorithm}

\begin{theorem}
	Algorithm~\ref{algo:opt} computes the optimal total cost solution in $O(2^mmn^2)$ time. Hence, it is fixed parameter tractable in $m$.
\end{theorem}
\begin{proof}
	It is easy to see that the running time of Algorithm~\ref{algo:opt} is $O(2^mmn^2)$. We argue for the proof of correctness. The base case is clear. The total cost is zero if there are no agents or no facilities. Suppose we want to compute $OPT(j,j',k,S)$ and we have already computed values for $OPT(h,h',\hat k,\hat S)$ where $h<j$ or $h'<j'$ or $\hat k<k$ or $\hat S\subset S$. We distinguish between two cases depending on whether $j-1$ is covered by the last facility in $S$.
	
	In the first case, $OPT(j,j',k,S)$ is equivalent to $OPT(j-1,j',k,S)$.
	In the second case we know that $j-1$ is not served by the last facility in $S$ but $j$ is. Hence, if $j$ is served by facility $s\in S$, then $OPT(j,j',k,S')$ is equivalent to $OPT(j-1,j-1,k-1,S\setminus \{s\}) + v([j,j'],s)$.
\end{proof}

% In a similar way, we use $OPT^{max}(j,j',k,S)$ to represent the optimal maximum cost when agents $[1,j']$ are partitioned into $k$ parts, and $j$ and $j'\ge j$ are in the $k$-th partition, and $S$ is the subset of $k$ facilities serving $[1,j']$.
% %\hau{could someone double check the below?}
% We use $v^{max}([j,j'],s)$ to denote the optimal maximum cost when agents $[j,j']$ are served by facility $s$ (with capacity $c$). We modify Algorithm~\ref{algo:opt} by replacing lines~\ref{line:part} with $s^*\gets\arg\min_{s\in S}\{\max\{OPT^{max}(j-1,j-1,k-1,S\setminus \{s\}),v^{max}([j,j'],s)\}\}$ and line~\ref{line:rec} with: $M[j,j',k,S]\gets \min\{OPT^{max}(j-1,j',k,S),
% max\{OPT^{max}(j-1,j-1,k-1,S\setminus\{s^*\}),v^{max}([j,j'],s^*)\}\}
% \}$.
%
% \begin{theorem}
% There is a $O(2^mmn^2)$ time algorithm to compute the optimal maximum cost. Hence, it is fixed parameter tractable in $m$.
% \end{theorem}

\paragraph{Maximum Cost.}

%Without loss of generality let the agents be ordered so that for every pair $i,i' \in N$, it holds that if $i \le i'$, then $x_i\le x_{i'}$. 
% Without loss of generality, let the facilities be ordered so that for every pair $j,j'$, it holds that if $j\le j'$, then $k_j\le k_{j'}$.
%We will use $[i,i']$ to denote the set of agents $\{i,\dots,i'\}$ for any $i\le i'$. 
%We compute the optimal way to locate $m$ facilities and partition
%agents into $m$ parts using dynamic programming. 
In a similar way, we use $OPT^{M}(i,i',j,S')$ to represent the optimal maximum cost when agents $[1,i']$ are partitioned into $j$ parts, and $i'\ge i$ are in the $j$-th partition, and $S'$ is the subset of $j$ facilities. We use $OPT^{M}(i,i',j,S')$ to denote optimal maximum cost when agents $[1,i']$ are partitioned into $j$ parts, and $i$ and $i'\ge i$ are in the $j$-th partition, and $S'$ is the subset of $j$ facilities.
% \begin{quote}
% 	$OPT^{max}(i,i',j,S')=$  optimal maximum cost when agents $[1,i']$ are partitioned into $j$ parts, and $i$ and $i'\ge i$ are in the $j$-th partition, and $S'$ is the subset of $j$ facilities.
% 	\end{quote}
We will make use of a $n\times n\times m\times 2^m$ array $E$, whose
entries are initially set to {\em empty}. We will use $E[i,i',j,S']$
to store $OPT^{M}(i,i',j,S')$. We use $v^{M}([i,i'],s_j)$ to
denote the optimal maximum cost when agents $[i,i']$ are covered by $s_j$ (with capacity $k_j$). 
%We denote $v^{max}([i,i'],s_j)$ to denote the optimal maximum cost for agents in $[i,i']$ when agents $[i,i']$ are covered by $s_j$ (with capacity $k_j$). 
% \begin{quote}
% $v^{max}([i,i'],s_j)=$ optimal maximum cost for agents in $[i,i']$ when agents $[i,i']$ are covered by $s_j$ (with capacity $k_j$).
% 	\end{quote}
% In order to find the optimal location for a single facility that is meant to serve a set of agents in interval $[i,i']$ we can simply try to put the location in the average of all pairs of agent locations. The optimal location will be the average of the extreme locations of agents who are actually served.
We invoke $OPT^{M}(n,n,m,S)$ to compute the optimal maximum cost and recover the solution from the values stored in $E$. 
The subroutine to solve OPT is summarized as Algorithm~\ref{algo:optmax}.
%Overall, the algorithm takes time $O(2^mmn^2)$.
\begin{algorithm}
\caption{$OPT^{M}(i,i',j,S')$}
\begin{algorithmic}[1]
		\footnotesize
\If{$i=0$ or $j=0$ or $S'=\emptyset$}
\State \Return $0$
\ElsIf{$E[i,i',j,S']$ not {\em empty}}
\State \Return $E[i,i',j,S']$
\Else
\State $E[i,i',j,S'] \gets \min(\max\{OPT^{M}(i-1,i',j,S'),\min_{s\in S'} \{\max(OPT^{M}(i-1,i-1,j-1,S\setminus \{s\}), v^{M}([i,i'],s))\}\})$
\State \Return $E[i,i',j]$
\EndIf
\end{algorithmic}
\label{algo:optmax}
\end{algorithm}
%\vspace{-1em}

\begin{theorem}
	Algorithm~\ref{algo:optmax} computes the optimal total cost solution in $O(2^mmn^2)$ time. Hence, it is fixed parameter tractable in $m$.
\end{theorem}

% \begin{theorem}
% 	There exists a $O(2^mmn^2)$ time algorithm to compute the
%         optimal maximum cost. Hence, the problem of finding a solution
%         with minimum maximum cost is fixed parameter tractable in $m$.
% 	\end{theorem}

\subsubsection{Identical capacities}

When facilities have the same capacity, 
we can use a different dynamic program to compute an optimal
solution with minimum total cost. 
We exploit the property that an optimal solution
consists of $m$ continuous regions served
by a facility located at the median agent of the
region.
We construct an array $C(i,j)$ which is the optimal
total cost supposing the leftmost
$i$ agents are served by the first $j$ facilities. 
Let $c$ be the capacity of the facilities. 
We initialize $C(i,j)$ with $\infty$. 
In the first round, we set $C(i,1)$ for $i$ in $1$ to $c$
with the optimal total cost for the first 
$i$ agents served by a single facility. 
In the $j$th round, for each $k$ from $j$
to $\mymin(n,jc)$ and $i$ from 1 to $\mymin(c,k-j+1)$ 
we update $C(k,j)$ with the minimum of $C(k,j)$ and
$C(k-i,j-1)$ plus the total cost of serving the $k-i$th
to the $k$th agents from left by the $j$th facility. 
In the final round, the optimal total cost will
be computed at $C(n,m)$. The runtime to
update each of the $O(nm)$ entries is
$O(c^2)$ giving a total cost of $O(nmc^2)$. 
Note that we can suppose $m < n$ otherwise
the optimal solution will have zero total cost. 
Similarly, we can suppose $c<n$. Hence, 
the total cost is polynomial in $n$. 
We can construct a similar dynamic program
to compute the optimal maximum cost, 
exploiting the fact that an optimal
solution consists in this case of $m$ continuous regions served
by a facility located at the midpoint between
the leftmost and rightmost agents served by the
facility. The total run time, in this case, is $O(nmc)$. 
%Since $c$ can be safely bounded by $n$, the running time can be written as $O(n^2m)$ which is faster than the alternative dynamic program of Brimberg et al.~\shortcite{BKEM01} which as a run time of $O(m^3n^2)$ for minimizing the total cost when facilities have identical capacities.
 These results complement those in Brimberg et al.~\shortcite{BKEM01}
 where an alternate dynamic program with run time $O(m^3n^2)$ is provided
 for minimizing the total cost when facilities have identical capacities.

\begin{theorem}
	There exists an $O(nmc^2)$ time algorithm to compute the
        optimal total cost and an $O(nmc)$ time algorithm to
compute the optimal maximum cost with $n$ agents
when all $m$ facilities have the same capacity $c$. 
	\end{theorem}

\section{A Mechanism Design Perspective}

In our mechanism design setting, each agent $j$'s position $x_j$ is her private information 
and a mechanism $M$ operates on the reported locations $(x'_1, \cdots, x_n')$ of all agents, 
which may be different from their actual locations.
Based on the reported locations, the mechanism $M(x_1',...,x_n')$ locates $m$ facilities, $y_1, ..., y_m \in {\mathbb R}$, 
and allocates disjoint sets of agents to each facility while respecting the capacity constraints 
(i.e., assign $a_j \in \{1, ..., m\}$ for each agent $j$ where $|\{ j \ | \ a_j = i\}| \le c_i$ for each facility $i$). 
The total and maximum costs of a mechanism are defined in the same way. 
Accordingly, every agent $j$'s cost is $u_j(x_j, M(x_1', ..., x_j', ..., x_n')) = |x_j - y_{a_j}|$
and they prefer mechanism outcomes that allocate them to a facility closer to their true location $x_j$. 
A mechanism is said to be \emph{strategyproof} if for every agent $j$ 
and for any profile of other agent reports, 
agent $j$ weakly prefers the mechanism outcome achieved by reporting her true location 
$x_j'=x_j$ to the outcome achieved by any other report $x_j''$. 
Our goal in this section is to find strategyproof mechanisms. 

When $m=1$ and $c_1 \ge n$, locating the facility at the median location
is strategyproof and optimal for minimizing the total cost. 
As a result, we focus on the case of two facilities. 
When $m=2$ and $c_1, c_2 \ge n$, 
Procaccia and Tennenholtz \shortcite{ptacmtec2013} first argue that no deterministic strategyproof mechanism
is guaranteed to return the optimal solution for either the total or the maximum cost.

In general, if we set the capacities of the locations high enough, we can inherit the lower bounds results from existing lower bound results. 
Particularly, \cite{lu2010asymptotically} proved that any deterministic strategyproof mechanism 
has an approximation ratio for the total cost of at least $\frac{n}{2}-1$,
which is further improved to be $n-2$ by \cite{ft2013}. 
Procaccia and Tennenholtz \shortcite{ptacmtec2013}  
proved that any deterministic strategyproof mechanism 
has an approximation ratio for the maximum cost of at least $2$.
An approximation ratio of $\alpha$ (w.r.t. some objective function to be minimized) implies that, across all instances of agent locations, the mechanism outcome evaluated according to the objective function is no more than a factor of $\alpha$ larger than the optimal outcome.

\subsection{Limitations of Existing Approaches}

The first contender is the \textbf{median mechanism} which has been studied for the case of one facility with unlimited capacity.
We can extend this mechanism to two or even more facilities with capacity constraints as follows: locate all facilities at the median agent. 
The mechanism remains strategyproof.
However, we can no longer bound the approximation ratio for either the total or maximum cost.
Suppose, for instance, most of the facilities are concentrated at the end points in the optimal solution.
We next explore the idea of placing facilities at the endpoints
with the \textbf{endpoint mechanism} which locates the facilities at the reported endpoints~\cite{PrTe09a}.
With two uncapacitated facilities, the endpoint mechanism has several nice properties. 
It is, for instance, strategyproof and returns a 2-approximation of the maximum cost, which is proved to be optimal by \cite{PrTe09a}.
Unfortunately, it is not hard to verify that, when we add (even uniform) capacities, the endpoint mechanism is not strategyproof
(i.e., agents can strategically exploit capacity limits to ensure a more favourable outcome).
The main difficulty here is how to allocate the agents to the facilities.

\subsection{Mechanisms for Two Facilities with Equal Capacity for Total and Maximum Costs}

We now present another mechanism for two capacitated facilities.
Suppose we have two facilities of capacities $c_1$ and $c_2$ where $n=c_1+c_2$. 
If we order agents from left to right, the \textbf{innerpoint mechanism}
places one facility at the $c_1$th agent from left, serving
the leftmost $c_1$ agents, and the other facility at the 
$c_1+1$th agent from the left, serving the rightmost
$c_2$ agents. If $c_1 \neq c_2$, we suppose
as before that the order of facilities
is fixed in advance of the agents reporting.
\begin{theorem}\label{theorem: innerpoint}
When $n$ is even and 2 facilities have equal capacity $n/2$,
the innerpoint mechanism is strategyproof 
and has a bounded approximation ratio of 2 for maximum cost and a ratio of $\frac{n}{2}-1$ for total cost.
\end{theorem}

\begin{proof}
%We omit the proof that the mechanism is strategyproof since this is implied by the more general statement in Theorem~\ref{theorem: SP rank}.
%
To show strategyproofness, we
consider three cases. In the first, consider one
of the $c_1-1$ leftmost agents. The only way they can
change the outcome is if they report a location to the
right of the $c_1$th agent. But this will give them a worst
outcome. In the second case, consider the $c_1$th agent
from the left. They have zero cost. Reporting any
other location leaves this the same or worse.  In the third
case, consider one of the $c_2$ rightmost agents.
This case is analogous to the  first two  cases. In
all cases, misreporting is not beneficial. 

To determine the approximation ratio for the maximum cost, we rescale the problem
so leftmost agent is at 0, and rightmost agent at 1. 
This does not affect the approximation ratio. By 
assumption, the two facilities have equal capacity $c$. 
Suppose the $c$th agent from left is at $x$, and
the $c+1$th agent is at $1-y$. The maximum cost is then
$\max(x,y)$. An optimal location of facilities would
put one facility at $\frac{x}{2}$, and the other at 
$1-\frac{y}{2}$. This has a maximum cost 
half as much of $\max(\frac{x}{2},\frac{y}{2})$. 
Hence, the approximation ratio is 2.

To determine the approximation ratio for the total cost,
we again rescale the problem so the leftmost agent is at 0, and the rightmost agent at 1.
%Again, by the assumption, the two facilities have equal capacity $c$. 
Suppose the $c$th agent from left is at $x$, and
the $c+1$th agent is at $x+u$. The worst case for 
the total cost has $c-1$ agents at 0, the $c$th agent at $x$, the 
$c+1$th agent at $x+u$ and the $c-1$ remaining
agents at 1. This gives a total cost of $(c-1)x + (c-1)(1-x-u)$.
That is, $(c-1)(1-u)$. 
The optimal solution, in this case, puts one facility at 0 and
the other at 1, giving a total cost of $x+(1-x-u)$.
That is, $1-u$. 
Hence, the approximation ratio for the total cost 
is at least $c-1$ (or $\frac{n}{2}-1$). 
%%{\color{red} The approximation ratio is, however, bounded.
%As before, if it was unbounded, we would require a sequence of problems on
%which the optimal total cost is zero or approaches zero, but 
%the innerpoint mechanism returns a solution that doesn't. 
%However, on such a sequence of problems, the solution
%returned by the innerpoint mechanism is zero or approaches zero. }
\end{proof}

Unfortunately, innerpoint no longer has a bounded approximation ratio for the 
total or maximum cost when facilities have
spare capacity i.e., $c_1+c_2>n$. The following theorem can be proved by constructing a suitable example.

\begin{theorem}
With 4 agents, and 2 capacitated facilities of size 3 or larger, 
the innerpoint mechanism has an unbounded approximation ratio for the total or maximum cost. 
\end{theorem}

\begin{proof}
Suppose two agents are at 0, one at $\epsilon$ and the fourth agent at 1. 
The innerpoint mechanism will locate one facility at 0, serving the two leftmost agents, and the other at $\epsilon$, serving the two rightmost agents. 
The total and maximum costs are both $1-\epsilon$ while the optimal total and maximum costs for two facilities with capacity 3 or 
greater are both $\epsilon$. 
The approximation ratio for total or maximum costs is therefore $\frac{1-\epsilon}{\epsilon}$ which tends to infinity as $\epsilon$ tends to zero.
\end{proof}

%\subsection{Characterizing Innerpoint within the wider class of Rank mechanisms}

Inspired by the median, endpoint and innerpoint mechanisms, 
we introduce a family of mechanisms that generalize all 
three. Given $m \ge 2$ facilities and $n$ agents, we let the capacity
of the $i$th  facility be $c_i$ and agent $j$ report location $x_j$
where $x_1 \leq \ldots \leq x_n$.  We suppose $n = \sum_{i=1}^m c_i$. 
The {\bf rank mechanism} has $m+1$ parameters,
$t_1$ to $t_m$ with $t_1 \leq \ldots \leq t_m$ and a permutation $\pi$ of $1$ to $m$. 
The mechanism locates facility $\pi(i)$ at $x_{t_i}$. 
Agents are then allocated to the facilities from left to right (while also respecting each facility's capacity constraint). 
If facilities have identical capacities, then we can ignore the permutation. 
We also note that our rank mechanism is similar to the percentile mechanisms introduced in \cite{sui2013analysis},
but the later only works for uncapacitated facilities.

We now provide a simple characterization of those
rank mechanisms which are strategyproof.

\begin{theorem}\label{theorem: SP rank}
The rank mechanism 
with parameters $t_1$ to $t_m$ and $\pi$ is 
strategyproof if and only if either $t_i=t_j$ for any $i$ and $j$, 
or there exists $1 \le k < m$ with  $t_1=t_k=\sum_{i=1}^kc_{\pi(i)}$ and
$t_{k+1} = t_m = t_1+1$. 
\end{theorem}
\begin{proof}
(Necessity)
There are two cases. In the first case,
$t_i=t_j$ for any $i$ and $j$.
All facilities are located at the
same location. An agent to the left of
this can only mis-report and move the
location of the facilities to the right which
is not in their interest. Similarly, an
agent to the right 
can only mis-report and move the
location of the facilities to the left which
is not in their interest. 
In the second case, 
there exists $1 \le k < m$ with $t_1=t_k=\sum_{i=1}^kc_{\pi(i)}$ and $t_{k+1} = t_m = t_1+1$. 
We can view this as equivalent to lumping the facilities into two ``super'' facilities, 
setting $m=2$, $t_1=c_{\pi(1)}$ and $t_2 = t_1+1$. 
This is the innerpoint mechanism which we have shown previously to be strategyproof.

(Sufficiency). 
There are again two cases. In the first  case, $m > 2$. 
We give a counter-example for strategyproofness for $m=3$. 
For greater $m$, we pad
the counter-example with additional facilities and
agents. We give a counter-example for
$t_1=t_2=1$ and $t_3=2$. Similar
counter-examples can be constructed for other
values of $t_i$. We suppose facilities have
capacity 2, and $x_1=0$, $x_2=3$, $x_3=4$, $x_4=5$, $x_5=6$, $x_6=7$. 
Agent 2 has an incentive to mis-report their position
as location 7. 
In the second case, $m=2$ and $t_1 \neq c_{\pi(i)}$. 
We give a counter-example for
$t_1=1$ and $t_2=2$. Similar
counter-examples can be constructed for other
values of $t_1$ and $t_2$. 
Suppose the two facilities
have capacity 2, 
and $x_1=0$, $x_2=3$, $x_3=4$, $x_4=5$. 
Agent 2 has an incentive to mis-report their
location to be position 5. 
\end{proof}

As an example, consider the {\bf quartile
mechanism} that places one facility at the first
quartile, and the second at the third quartile.
% When there is no spare capacity, this is optimal for minimizing the total cost. If the capacity of facilities $c$ is even, the quartile mechanism is
% an instance of the rank mechanism
% with $t_1 = \frac{c}{2}$ and $t_2 = \frac{3c}{2}$.
It follows immediately
from the last theorem that the quartile mechanism is not strategyproof.
%chanisms which are strategyproof.

We now provide a simple necessary condition for the rank mechanism to have bounded approximation ratio.

\begin{theorem}\label{theorem: bounded rank}
With 2 (or more) facilities with capacity constraints, if the rank mechanism has bounded approximation ratio for the total and maximum cost it must be that $t_1\le c_1$ and $t_m\ge c_1+1$ (recall that $t_1\le t_m$), where $c_1$ is the capacity of facility $\pi(1)$.
\end{theorem}
\begin{proof}
Consider an instance where $c_1$ agents located at 0 and the remaining $n-c_1=\sum_{j=2}^m c_j$ agents are located at $1$. 
The optimal total and maximum costs are zero. 
If either $t_1>c_1$ or $t_m<c_1+1$ holds, since $t_j\le t_k$ for all $j\le k$, all facilities will be located at the same location (either $0$ or $1$). 
This necessarily gives non-zero total and maximum costs 
and hence unbounded approximation ratios. 
\end{proof}

% Theorems \ref{theorem: innerpoint} to \ref{theorem: bounded rank}  are proved in the supplementary materials. 
Theorem~\ref{theorem: innerpoint}, Theorem~\ref{theorem: SP rank}, and Theorem~\ref{theorem: bounded rank} collectively lead to the following theorem which characterizes the innerpoint mechanism as the only strategyproof rank mechanism with bounded approximation ratio (for either total or maximum cost) when there are two facilities of equal capacities.

\begin{theorem}
With 2 capacitated facilities of equal capacity, the only rank mechanism which is both strategyproof and provides a bounded approximation ratio for both the maximum and total cost is the innerpoint mechanism, i.e, $t_1=c_1$ and $t_2=c_1+1$.
\end{theorem}

In summary, the mechanisms in this section are either unbounded or not strategyproof or only satisfy the desirable properties if $c_1=c_2=n/2$. % (see Table~\ref{table:lit}). 
Next we present a new mechanism that is strategyproof and provides bounded approximations even for two capacitated facilities and the case where there may be spare capacities.

% \begin{table}[h!]
% \begin{center}
% 	  \scalebox{0.7}{
%     \begin{tabular}{ |c|cc|cc|}
%     \hline
%  & Total Cost & & Maximum  Cost&\\ \hline
%  & Lower & Upper  & Lower  & Upper \\ \hline
% % All SP mechanisms &3  &$\infty$ &2 &$\infty$ \\ \hline
% Median (SP) &$\infty$  &$\infty$ &$\infty$ &$\infty$ \\ \hline
% Endpoint (not SP) &$\frac{n}{2}-1$  &bounded & 1&2\\ \hline
% Innerpoint (SP) & 1 & $n-2$& 1&2 \\ \hline
%      \end{tabular}
%      }
% \end{center}
% \caption{Approximation bounds of mechanisms for 2 capacitated facilities with capacities summing up to $n$. The lower bound results hold even for two facilities with equal capacity summing to $n$. %The upper bound results for innerpoint only hold when there is no spare capacity.
% }
% \label{table:lit}
% \end{table}

\subsection{Mechanisms for Two Facilities with Arbitrary Capacities for Total and Maximum Costs}
%\section{A New Mechanism for Two Facilities}
In this section, we present a more general strategyproof mechanism that 
almost matches the existing lower bounds of \cite{lu2010asymptotically} and \cite{ft2013}. 
Let $x=(x_{i})_{i\in N}$ be the location profile such that $x_{1}\leq \cdots \leq x_{n}$ 
and $c_{1} %\geq \frac{n}{2} 
\geq c_{2} \geq 1$ be the capacities such that $c_{1}+c_{2} \geq n$. 
%Set $l_{1}=n-c_{2}$ and $l_{2}=n-c_{1}$.
Let $f_{1}$ and $f_{2}$ be the output locations of the two facilities.

\paragraph{Extended Endpoint Mechanism (EEM)} 
Let $X_{1}=\{i | x_{i}-x_{1}\leq \frac{1}{2}(x_{n}-x_{1})\}$ and $X_{2}=\{i | x_{n}-x_{i}<\frac{1}{2}(x_{n}-x_{1})\}$.
{If $|X_{1}|\geq |X_{2}|$, execute one of the following three cases.}
% \haris{WLOG is fine in a proof but it will be nice if we can somehow compactly describe the mechanism including the other case. If not, please write out the other case even if we don't prove it. }
\begin{description}
\item[Case 1.] If $ |X_{1}| \leq c_{1}$ and $|X_{2}|\leq c_{2}$, $f_{1}=x_{1}$ and $f_{2}=x_{n}$. 
Agents in $X_{1}$ are allocated to $f_{1}$ and the others are allocated to $f_{2}$.
\item[Case 2.] If $ |X_{1}| > c_{1}$ and $|X_{2}| \leq c_{2}$, $f_{1}=2x_{c_{1}+1}-x_{n}$ and $f_{2}=x_{n}$. Agents $\{1,\cdots, c_{1}\}$ are allocated to $f_{1}$ and the others are allocated to $f_{2}$.
\item[Case 3.] If $ |X_{1}| \leq c_{1}$ and $|X_{2}|> c_{2}$, $f_{1}=x_{1}$ and $f_{2}=2x_{n-c_{2}}-x_{1}$. Agents in $\{1,\cdots,n-c_{2}\}$ are allocated to $f_{1}$ and the others are allocated to $f_{2}$.
\end{description}
{If $|X_{1}| < |X_{2}|$, switch the roles of the two facilities in above cases and execute one of them.}

EEM is essentially an endpoint mechanism, 
but to restore strategyproofness, it locates one facility outside of $[x_{1},x_{n}]$.
Although EEM has slightly poorer performance than the classical endpoint mechanism,
with respect to the lower bound results, EEM is actually optimal up to a constant.

\begin{theorem}
EEM is strategyproof.
\end{theorem}

\begin{proof}
Without loss of generality, assume $|X_{1}|\geq |X_{2}|$ as the other case is symmetric.
{Note that %no matter which case happens 
$f_{k}\leq x_{1}$ and $f_{2-k}\geq x_{n}$ always hold for some $k\in\{1,2\}$.

For Case 1 where $|X_{1}| \leq c_{1}$ and $|X_{2}| \leq c_{2}$, 
each agent is served by the facility located at the nearest endpoint to her,
and the two endpoint agents get the best possible solution.
If any agent except the two endpoint agents deviates and the resulting locations of the two facilities change,
she can only be worse off as under no situation she could be served by a facility located in $(x_{1}, x_{n})$.
Thus our mechanism is strategyproof in this case.}

Next we consider Case 2 where $ |X_{1}| > c_{1}$ and $|X_{2}| \leq c_{2}$.
\begin{itemize}
\item For any agent $1\leq i \leq c_{1}$,  $i \in X_{1}$ and is served by the facility located at $f_1=2x_{c_1+1}-x_n$. 

Under EEM, any deviation by agent $i$ to $x_i'< x_{c_1+1}$ will continue to allocate agent $i$ to the left most facility whose location $f_1'$ either remains unchanged at $2x_{c_1+1}-x_n$, or, if Case 1 occurs, is moved to $f_1'=x_i'<x_1$. It is straightforward to see that neither of these cases is profitable for the agent, and so we restrict our attention to deviations $x_i'\ge x_{c_1+1}$. Let $X_1'$ and $X_2'$ be the new partitions of the EEM. If $X_1', X_2'$ are such that Case 1 or Case 2 holds then agent $i$ is now allocated to $f_2'=\max\{x_n, x_i'\}$, this is never a profitable deviation.  Finally, we note that Case 3 can never occur  since this would require that $|X_1|=c_1+1$ and  $|X_2|=c_2$ which contradicts $|X_1|+|X_2|=n$ and $c_1+c_2\ge n$. Thus $i$ cannot be better off by deviating.

%if she reports some location on the left of $x_{c_{1}+1}$, 
%it can only make $f_{1}$ to move left or no change at al and she is still connected with $f_{1}$, which is not better off.
%If she reports some location on right of $x_{c_{1}+1}$, 
%then agent $c_{1}+1$ is connected to $f_{1}$ and $i$ is connected to $f_{2} \geq x_{n}$,
%no matter which case happens after the deviation.
%However, $f_{2}-x_{i}\geq x_{n}-x_{c_{1}+1}=x_{c_{1}+1}-f_{1}\geq x_{i}-f_{1}$,
%thus $i$ does not get better off.

\item Next, we consider agent $i \geq c_{1}+1$. 
If she reports some location on the left of $x_{c_{1}}$, then agent $c_{1}$ will replace the role of $c_{1}+1$ in the mechanism 
thus $f'_{1}=2x_{c_{1}}-x_{n}$ and $i$ is connected to $f'_{1}$.
However, $x_{i}-f'_{1} \geq x_{c_{1}}-f'_{1} =x_{n}-x_{c_{1}} \geq x_{n} - x_{i}$,
thus $i$ does not profit from the misreporting their location.
Note that $i$'s deviation cannot cause the number of agents in $X_1$ to be less than the number of agents in $X_2$.
Otherwise, $|X_{1}| = |X_{2}| + 1$, then $c_{1}$ must be at least as large as $|X_{1}|$ as $c_{1} \geq c_{2}$ and $c_{1} + c_{2} \geq n$,
Thus Case 1 must happen when $i$ tells the truth.
Accordingly, if $i$ reports some location on the right of $x_{c_{1}}$, 
she is still connected to $f_{2}$ (or even the right of $f_{2}$), no matter which case happens after the deviation.
Thus $i$ never benefits from misreporting.
\end{itemize}

By a similar argument as for Case 2, EEM is strategyproof under Case 3.
\end{proof}

\begin{theorem} \label{lem:eem:social}
For the objective of total cost, EEM has an approximation ratio of $\frac{3n}{2}$.
\end{theorem}

\begin{proof}
Again, without loss of generality, assume $|X_{1}|\geq |X_{2}|$. % and the other case is symmetric.
We first note that the optimal solution has the following form: 
partition of the agents into left successive $c_{1}$ agents and right successive $c_{2}$ successive agents
and locate each facility at the median point of each subset of the agents.

For Case 1, EEM is exactly the same as the classical Endpoint
Mechanism \cite{PrTe09a}, 
thus has an approximation of $n-2$.

For Case 2, $OPT\geq x_{n}-x_{c_{1}+1}$.
This is because in the optimal solution, all the agents in $X_{1}$ cannot be served by a single facility as $|X_{1}| > c_{1} \geq c_{2}$.
Thus at least one of the agents in $X_{1}$ has to be grouped with $x_{n}$.
Moreover, 
\begin{align*}
ALG & \leq  c_{1}(x_{c_{1}}-f_{1}) + (n-c_{1})(f_{2}-x_{c_{1}+1}) \\
	& \leq  c_{1}(x_{c_{1}+1}-f_{1}) + (n-c_{1})(f_{2}-x_{c_{1}+1}) \\
	& =  n(f_{2}-x_{c_{1}+1}) \leq nOPT.
\end{align*}

For Case 3, $OPT\geq \frac{1}{2}(x_{n}-x_{1})$.
This is because in the optimal solution $X_{1}$ 
cannot be served by one facility and $X_{2}$ by another facility.
Thus either at least one of the agents in $X_{1}$ has to be grouped with $x_{n}$
or at least one of the agents in $X_{2}$ has to be grouped with $x_{1}$.
No matter which case happens, 
there is one group such that the longest distance between each pair of agents of this group is at least $\frac{1}{2}(x_{n}-x_{1})$.
Moreover, 
\begin{align*}
ALG & \leq  |X_{1}|(x_{|X_{1}|}-f_{1}) + (n-c_{2} - |X_{1}|)(x_{n-c_{2}}-f_{1}) \\
        & ~~~~~ + c_{2}(f_{2}-x_{n - c_2+1}) \\
	&  \leq |X_{1}|\cdot \frac{1}{2} (x_{n}-x_{1}) + (n-|X_{1}|)\cdot  (x_{n}-x_{1})  \\
	& \leq \frac{3n}{4}(x_{n}-x_{1}) \leq \frac{3n}{2} OPT.
\end{align*}
The second inequality is because $f_{2}-x_{n - c_2+1} \leq f_{2}-x_{n - c_2} = x_{n-c_{2}}-x_{1} \leq x_{n}-x_{1}$.
The third inequality is because the term in the second line is maximized when $|X_{1}| = \frac{n}{2}$ as $|X_{1}| \geq |X_{2}|$.
% \hau{not clear .... can you just not that like case 2 but you have  $OPT\geq x_{c_{n-2 }}-x_{1}$}
%
% \bo{Your bound is not correct.
% Consider the case: $x_{1}=0, x_{2}=0.5, x_{3}= x_{4}=1$ and $c_{1}= 3, c_{2}=1$. So $OPT = \frac{1}{2}(x_{4}-x_{1}) < x_{n-c_{2}} - x_{1}= 1$.}
\end{proof}

%We do not know if the analysis in Lemma \ref{lem:eem:social} is tight or not.

%The analysis in Lemma \ref{lem:eem:social} is asymptotically tight.
%Consider the example:
%$x_{1}=0$, 
%$x_{i}=0.5-\epsilon$ for $1<i\leq k+1$,
%$x_{i}=0.5+\epsilon$ for $k+1<i\leq 2k$,
%$x_{2k+1}=x_{2k+2} = 1$
%and $c_{1} = 2k$, $c_{2}=2$.
%Note that the optimal solution is $(f_{1},f_{2}) = (0.5, 1)$ and $OPT\approx 0.5$.

\begin{theorem}\label{lem:eem:max}
For the objective of maximum cost, EEM has an approximation ratio of $4$.
\end{theorem}

\begin{proof}
%Again, without loss of generality, assume $|X_{1}|\geq |X_{2}|$ and the other case is symmetric.
%We first note that the optimal solution still has the following form: 
%partition of the agents into left successive $k_{1}$ agents and right successive $k_{2}$ successive agents
%and locate each facility at the {\em center} point of each subset of the agents.
For Case 1, EEM is exactly the same as the classical Endpoint Mechanism \cite{PrTe09a}, thus has an approximation of $2$.

% \hau{don't completely understand Case 2 and 3; can you please rewrite and expand carefully?}

For Case 2, let $dis=x_{n}-x_{c_{1}+1}$.
We claim that $OPT = \frac{1}{2}dis$.
We first note that $(f^{*}_{1},f^{*}_{2})=(\frac{1}{2}(x_{c_{1}}-x_{1}), \frac{1}{2}(x_{n}+x_{c_{1}+1}))$ 
is a feasible solution with the announced cost,
i.e. $(f^{*}_{1},f^{*}_{2})$ an optimal solution.
Let $(l_{1},l_{2})$ be any solution.
It is easy to see that to guarantee a cost that is not greater than $\frac{1}{2}dis$,
$\max\{l_{1}, l_{2}\} \geq \frac{1}{2}(x_{n}+x_{c_{1}+1})$
as agent $n$ has to be connected to one of them.
Moreover, since the capacity of the leftmost facility is at most $c_{1}$,
at least one of agents $\{1,\cdots,c_{1},c_{1}+1\}$ has to be grouped with agent $n$, 
thus the optimal solution is to put $c_{1}+1$ into $n$'s group and 
its serving facility is located at $\frac{1}{2}(x_{n}+x_{c_{1}+1})$.
We see that the cost of $(l_{1},l_{2})$ cannot be smaller than $\frac{1}{2}dis$.

Let $(f_{1}, f_{2})$ be the output of EEM.
Then $x_{c_{1}}-f_{1} \leq x_{c_{1}+1}-f_{1}=x_{n}-x_{c_{1}+1}=dis$.
Thus, $ALG\leq dis = 2OPT$, which is a 2-approximation.

Case 3 is similar to Theorem~\ref{lem:eem:social}. 
Since the optimal solution cannot serve $X_{1}$ 
by one facility and $X_{2}$ by another facility,
either one of the agents in $X_{1}$ has to be grouped with $x_{n}$
or one of the agents in $X_{2}$ has to be grouped with $x_{1}$.
Thus $OPT\geq \frac{1}{4}(x_{n}-x_{1})$.
Again, since $f_{2}-x_{n - c_2+1} \leq f_{2}-x_{n - c_2} = x_{n-c_{2}}-x_{1} \leq x_{n}-x_{1}$.
$ALG\leq 4\ OPT$, which is a 4-approximation.
% \hau{Didn't even check; can you fix the other cases first?}
\end{proof}

\section{Conclusion and Discussion}

%In this paper we considered a new facility location problem where facilities are capacity constrained. 
%\begin{table}[h!]
%\begin{center}
%	  \scalebox{0.7}{
%    \begin{tabular}{ |c|cc|cc|}
%    \hline
% & Total Cost & & Maximum  Cost&\\ \hline
% & Lower & Upper  & Lower  & Upper \\ \hline
% All SP mechanisms & $\frac{n}{2}-1$&$\infty$ &2 & $\infty$\\ \hline
% EEM & $n-2$&$\frac{3n}{2}$ &2 & 4\\ \hline
%
%     \end{tabular}
%     }
%\end{center}
%\caption{Approximation bounds of mechanisms for 2 capacitated facilities when allowing for spare capacity. The lower bound of $n-2$ is from \cite{ft2013}.}
%\label{table:lit}
%\end{table}

We considered the FLP-CC both algorithmically 
and from a mechanism design perspective. 
See Table~\ref{table:results} for a summary of our results. 
%(see Table~\ref{table:lit} for mechanism design results). 
Collectively, our results show that the addition of capacity
constraints to the FLP makes it more difficult 
compared to the uncapacitated case to solve optimally and to design
mechanisms with desirable properties (e.g.,  
strategyproofness and bounds on the approximation
ratio). % of the total and maximum cost). 
%Nevertheless, with a little effort, 
% We were able to design a mechanism
% which is strategyproof and returns a solution
% that is within a factor of 2
% of the optimal maximum cost.
There are many directions for future work.
The most important question is to understand whether any meaningful 
upper bound can be established for the case of three locations. %We note here that even if locations have infinite capacity, the case of three locations has remained elusive since the seminar paper on this topic~\cite{PrTe09a,ptacmtec2013} and has been termed the `\emph{`most important technical problem}.''
%Other research questions include the following. 
Moreover, can we extend results beyond one dimension to trees, networks, or two-dimensional rectilinear and Euclidean metrics? 
% Can we design a
% strategyproof mechanism with
% a bound on the approximation
% ratio of the total cost that is a small
% constant? This would close the gap
% between the lower bound of 3 identified here
% and the linear lower bound for
% the innerpoint mechanism.
Finally, it is interesting to consider if randomization can help to design 
more efficient and better algorithms or mechanisms for FLP-CC. 

\section*{Acknowledgements}
Aziz is supported by a UNSW Scientia Fellowship, and Defence Science and Technology (DST) under the project ``Auctioning for distributed multi vehicle planning'' (DST 9190).
 Lee is supported by a Data61 and a UNSW Scientia PhD fellowship.
Li is supported by the ERC grant number 639945 (ACCORD). 
Walsh is funded by the ERC under Horizon 2020 via
AMPLify 670077.

\bibliographystyle{aaai}
\bibliography{capFl_arxiv.bib}

\end{document}